\documentclass[10pt,a4paper]{article}

\usepackage[a4paper,text={160mm,250mm},centering,headsep=10mm,footskip=15mm]{geometry}
\usepackage[T1]{fontenc}
\usepackage{graphicx}
\usepackage{epsfig, float}
\usepackage{amsmath,amssymb}
\usepackage{amsfonts}
\usepackage{url}
\usepackage{bbm}
\usepackage{mathrsfs} 
\usepackage[english]{babel}
\usepackage{hyperref}
\usepackage{cite}

\newcommand{\R}{\mathbb{R}}
\newcommand{\CC}{\mathbb{C}}

\newcommand{\vx}{\mathbf{x}}

\newcommand{\free}{{\rm free}}
\newcommand{\ret}{{\rm ret}}
\newcommand{\sym}{{\rm sym}}

\newcommand{\be}{\begin{equation}}
\newcommand{\ee}{\end{equation}}

\newtheorem{theorem}{Theorem}[section]

\newtheorem{proposition}[theorem]{Proposition}

\newenvironment{proof}[1][Proof:]{\begin{trivlist}
\item[\hskip \labelsep {\bfseries #1}]}{\end{trivlist}}

\newcommand{\qed}{\hfill\ensuremath{\square}}

\title{Probability conservation for multi-time integral equations}

\author{
Matthias Lienert\thanks{Marvel Fusion GmbH, Theresienh\"ohe 12, 80339 Munich, Germany. E-mail: lienertmat@gmail.com}
}

\date{November 20, 2022}

\begin{document}

\maketitle

\begin{abstract}
\noindent In relativistic quantum theory, one sometimes considers integral equations for a wave function $\psi(x_1,x_2)$ depending on two space-time points for two particles. A serious issue with such equations is that, typically, the spatial integral over $|\psi|^2$ is not conserved in time -- which conflicts with the basic probabilistic interpretation of quantum theory.
However, here it is shown that for a special class of integral equations with retarded interactions along light cones, the global probability integral is, indeed, conserved on all Cauchy surfaces. For another class of integral equations with more general interaction kernels, asymptotic probability conservation from $t=-\infty$ to $t=+\infty$ is shown to hold true. Moreover, a certain local conservation law is deduced from the first result.
\end{abstract}


\begin{center}
\emph{This paper is dedicated to the memory of Detlef D\"urr,\\
a wonderful person, scientist and mentor.}
\end{center}

\section{Introduction} \label{sec:introduction}

\subsection{Motivation} \label{sec:motivation}

An elegant but little-known approach to relativistic quantum theory involves wave functions $\psi(x_1,...,x_N)$ depending on many space-time variables $x_k \in \R^4$ for many particles $k=1,2,...,N$. In addition to other applications, such as by Detlef D\"urr and coauthors \cite{berndl,hbd} in the foundations of relativistic quantum theory, these \textit{multi-time wave functions} (see \cite{multitime_book} for an introduction) make it possible write down closed integral equations which describe a fixed number of relativistic, interacting particles in a manifestly covariant way. The best-known example is the Bethe-Salpeter (BS) equation \cite{bs_equation} which has been used to describe bound states in quantum field theory (QFT). At the time of its discovery, it was hoped that the BS equation represented a fully relativistic -- and interacting -- generalization of the Schr\"odinger equation, at least for processes where fermion creation and annihilation are not relevant (such as bound state problems).

\begin{quote}
The formulation of a completely relativistic wave
equation for two-body systems has, in a certain
sense, solved a long-standing problem of quantum
mechanics. The natural and simple way in which
relativistic invariance is achieved is, of course, very real
progress, which may lead one to hope that the main
features of the equation are more permanent than the
solidity of its present field theoretic foundation might
suggest. Furthermore, it is hardly necessary to recall
that the usefulness of the equation has been amply
demonstrated in several high-precision calculations of
energy levels. -- Wick, 1954 \cite{wick_54}
\end{quote}

However, there is a serious issue with such integral equations. Due to their non-locality
in configuration space (in the sense of PDEs), they typically do not imply (local) continuity equations, nor do they conserve the (global) probability integral. In the context of the Bethe-Salpeter equation, it has been said:

\begin{quote}
[...] The absence of a positive-definite norm for the
wave function and of any orthogonality theorem. -- Wick, 1954, listing problems of the BS equation \cite{wick_54}

Nakanishi (1965) explicitly calculated the normalization integrals in some special cases of the equal-mass Wick-Cutkosky model, and discovered that certain B-S amplitudes have negative or zero norm. -- Nakanishi, 1969 \cite{nakanishi}
\end{quote}

Of course, these quotes mean nothing else than that the quantities proposed as a norm do not actually constitute one. They cannot have the physical meaning of a probability integral. Considering that quantum physics is based on the notion of probability, this seems rather problematic for the physical justification of the Bethe-Salpeter equation.\\

The motivation for integral equations for a multi-time wave function can also be approached from a second angle -- one that was dear to Detlef D\"urr: the quantization of Wheeler-Feynman (WF) electrodynamics \cite{schwarzschild,tetrode,fokker_1,wf1,wf2}; see \cite{existence_wf,ddv_2014,bddh_2017} for some of Detlef D\"urr's works on the topic. This theory pursues the idea that the ultraviolet divergence problem of classical Maxwell-Lorenz electrodynamics can be avoided by ``integrating out'' the fields. The result is a dynamics where interactions between particles occur directly and exactly when particle world-lines are light-like separated.

The discussions with Detlef D\"urr about finding a suitable quantum version of that theory sparked my personal curiosity about the subject. After studying previous proposals for quantizations of WF electrodynamics \cite{davies_1970,hoyle_narlikar,barut_multiparticle, ludwig_50}, \cite[chap. 8]{deckert_phd}, which all encounter their own difficulties, it seemed to me that integral equations for a multi-time wave function might be a more promising way forward. In \cite{direct_interaction_quantum}, I laid out how these types of integral equations make it possible to transfer the principle of direct interactions along light cones, that forms the core of WF electrodynamics, to the quantum level. This was done in a way that retains the Dirac-Schr\"odinger equation with a spin-dependent Coulomb potential as the non-relativistic limit, thus staying close to empirically successful models. In a series of papers \cite{mtve,int_eq_curved,int_eq_dirac,int_eq_singular}, my co-authors and me were able to prove that multi-time integral equations provide a fully relativistic and interacting quantum dynamics which does not suffer from the ultraviolet divergence problem, even for singular light-cone interactions \cite{int_eq_singular}.

However, the question of probability conservation was left open and, as we have seen for the BS equation, there is reason for concern. Equations with interaction terms which are non-local in the configuration space of quantum mechanics do typically not imply local conservation laws. It is then at best unclear whether global probability conservation holds true. Historically, Feynman himself saw the problem of probability conservation as one of the central obstacles to quantizing WF electrodynamics, as he reports in his Nobel lecture:

\begin{quote}
I found that if one generalized the action from the nice Lagrangian forms [...] to these forms [...] then the quantities which I defined as energy, and so on, would be complex. The energy values of stationary states wouldn’t be real and probabilities of events wouldn’t add up to 100\%. -- Feynman, 1965 \cite{feynman_nobel_lecture}

I don’t think we have a completely satisfactory relativistic quantum-mechanical model, even one that doesn’t agree with nature, but, at least, agrees with the logic that the sum of probability of all alternatives has to be 100\%. -- Feynman, 1965 \cite{feynman_nobel_lecture}
\end{quote}

In view of these difficulties, if one is not ready to dismiss multi-time integral equations altogether, one may conclude that the equations are not exactly the right ones and that some modification is in order. From the point of view that the interaction term in the BS equation which quantum electrodynamics suggests (an infinite series of Feynman diagrams in need of renormalization) is not the most simple and natural, such a modification seems easy to accept. In addition, the argument that non-local interaction terms usually preclude local conservation laws does not apply to global conservation laws, leaving room for logical possibilities which may not have been sufficiently explored.
I am going to adopt these positions here. This makes it possible to prove that for certain classes of integral equations the global probability integral is, in fact, conserved.

\section{The integral equation} \label{sec:int_eq}

For simplicity, we focus on the case of $N=2$ Dirac particles. Moreover, we set $c = 1 = \hbar$. Then the class of integral equations we shall study reads:
\be
\psi(x_1,x_2) = \psi^{\rm free}(x_1,x_2) + i \int d^4 x_1' \, d^4 x_2' \, G_1(x_1-x_1') G_2(x_2-x_2') K(x_1',x_2') \psi(x_1',x_2').
\label{eq:inteq}
\ee
Here $\psi : \R^4 \times \R^4 \rightarrow \CC^4 \otimes \CC^4$ is a multi-time wave function with 16 complex components for two particles.\footnote{One can also study \eqref{eq:inteq} on a sub-domain of $\R^4\times \R^4$, e.g., on the set $\mathscr{S}$ of space-like configurations.}
$\psi^{\rm free}(x_1,x_2)$ is a solution of the free Dirac equation in $x_1,x_2$, i.e.:
\be
(-i \gamma_k^\mu \partial_{x_k^\mu} +m_k)\psi^{\rm free}(x_1,x_2) = 0,~~k=1,2.
\label{eq:multitimedirac}
\ee
The space-time integral in \eqref{eq:inteq} extends over $\R^4 \times \R^4$, the entire configuration space-time.
$G_1$ and $G_2$ are Green’s functions of the Dirac equations of particles 1 and 2. We use the convention of \cite[Appendix E]{birula_qed}:
\be
(-i \gamma_k^\mu \partial_{x_k^\mu} +m_k)G_k(x_k-x'_k) = \delta^{(4)}(x_k-x'_k).
\ee
Here and in the following, particle indices in $\gamma$-matrices, Green’s functions and propagators indicate on which spin index their matrix structure acts.

$K(x_1,x_2)$ is the so-called \emph{interaction} kernel, a covariant, matrix-valued distribution. We require the following symmetry condition with respect to its matrix structure:
\be
K^\dagger(x_1,x_2) = \gamma_1^0 \gamma_2^0 K(x_1,x_2) \gamma_1^0 \gamma_2^0.
\label{eq:symmetryk}
\ee
As explained in \cite{direct_interaction_quantum}, direct interactions along light cones in the spirit of Wheeler-Feynman electromagnetism can be expressed by the interaction kernel
\be
K^\sym(x_1,x_2) = \lambda \, \gamma_1^\mu \gamma_{2,\mu}\, \delta((x_1-x_2)^2)
\label{eq:lightconeint}
\ee
where $\lambda \in \R$ is a coupling constant and $(x_1-x_2)^2 = (x_1^0 - x_2^0)^2 - |\vx_1-\vx_2|^2$ denotes the Minkowski square. Note that \eqref{eq:lightconeint} contains both retarded and advanced interaction terms, as can be seen by decomposing the delta distribution. The retarded part is given by:
\be
K^\ret(x_1,x_2) = \lambda \, \gamma_1^\mu \gamma_{2,\mu}\, \frac{1}{2 |\vx_1-\vx_2|} \delta(x_1^0 - x_2^0 - |\vx_1-\vx_2|).
\label{eq:retardedint}
\ee

With these conventions, the factor $i$ in the interaction term in \eqref{eq:inteq} is required to obtain the correct non-retarded limit, i.e., a Schr\"odinger equation with spin-dependent Coulomb potential\cite{direct_interaction_quantum}.

\paragraph{Relation to the Bethe-Salpeter equation.}
The BS equation is contained in the class of equations \eqref{eq:inteq} for the case that $G_k$ are Feynman propagators $S^F_k$ for the two particles $k=1,2$, and for the case that $K(x_1,x_2)$ is given by an infinite series of Feynman diagrams. In the so-called ladder-approximation of the BS equation, only a certain sub-class of these Feynman diagrams (consisting of those exchanging only one virtual photon at a time) is considered. Then $K$ simplifies to
\be
K^\text{BSL}(x_1,x_2) = \lambda \, \gamma_1^\mu \gamma_2^\nu D^F_{\mu \nu}(x_1,x_2)
\label{eq:bslinteraction}
\ee
where $D^F_{\mu \nu}$ is the Feynman propagator of a photon (see \cite[p. 331]{greiner_qed}). In Lorenz gauge:
\be
D^F_{\mu \nu}(x_1,x_2) = \eta_{\mu \nu} \, D^F(x_1,x_2)
\label{eq:df}
\ee
where $D^F$ is the Feynman propagator of the wave eq. and $\eta_{\mu \nu}$ the Minkowski metric. As both $D^F$ and $\frac{1}{4\pi}\delta((x_1-x_2)^2)$ are Green’s functions of the wave equation, \eqref{eq:bslinteraction} closely resembles \eqref{eq:lightconeint}. However, a crucial difference is that only \eqref{eq:lightconeint} is supported on the light cone; \eqref{eq:bslinteraction} also has support outside.

\paragraph{Role of the dynamics.}
As discussed in \cite{int_eq_dirac}, one can best understand the dynamics defined by \eqref{eq:inteq} in the case of $G_k = S_k^\ret$, the retarded Green’s function of the Dirac equation for particle $k$.
Then, for each incoming free wave function $\psi^\free$, the integral equation defines a unique interacting solution $\psi$ which agrees with $\psi^\free$ in the infinite past. Thus Eq. \eqref{eq:inteq} can be viewed as a machinery which takes an incoming free solution and computes an interacting correction to it.

\paragraph{Notes on retarded Green’s functions.}
We now collect useful properties of retarded Green’s functions which are rooted in their simple relation to the propagator of the Dirac equation. These will play a crucial role in the upcoming arguments. Namely, we have:
\be
S^\ret(x-x’) = \theta(x^0-{x’}^0)S(x-x’)
\label{eq:sret}
\ee
where $\theta$ is the Heaviside function and $S$ the propagator of the Dirac equation. $S$ can be used to time-evolve every free solution of the Dirac equation from one Cauchy surface $\Sigma$ to another:
\be
\psi^\free(x) = -i \int_\Sigma d \sigma_\mu(x') \, S(x-x') \gamma^\mu \psi^\free(x').
\label{eq:freeevol}
\ee
Confusingly, both $S^\ret$ and $S^F$ are called propagators, even though they do not have the property \eqref{eq:freeevol} for all wave functions and all Cauchy surfaces.
From \eqref{eq:freeevol}, one can deduce the composition property
\be
\int_\Sigma d \sigma_\mu(x') \, S(x-x')\gamma^\mu S(x'-x'') = i S(x-x'').
\label{eq:compositionproperty}
\ee
Moreover, we have $[S(x-x')]^\dagger = -\gamma^0 S(x'-x) \gamma^0$.  
This allows us to compute the adjoint of the integral equation \eqref{eq:inteq}, denoting $\overline{\psi}(x_1,x_2) = \psi^\dagger(x_1,x_2) \gamma_1^0 \gamma_2^0$:
\begin{align}
\overline{\psi}(x_1,x_2) = \overline{\psi}^{\rm free}(x_1,x_2) - i \int d^4 x_1' \, d^4 x_2' \, \overline{\psi}(x_1',x_2') &K(x_2',x_1') S_1(x_1'-x_1) S_2(x_2'-x_2)\nonumber\\
&\times \theta(x_1^0 - {x_1'}^0) \theta(x_2^0 - {x_2'}^0).
\label{eq:inteq2}
\end{align}

\section{Relativistic probability conservation} \label{sec:prob_cons}

For Dirac particles, local probability conservation is expressed by the continuity equation
$\partial_{x^\mu} j^\mu(x) = 0$ where $j^\mu= \overline{\psi}(x) \gamma^\mu \psi(x)$ denotes the probability current and $\overline{\psi}(x)= \psi^\dagger(x) \gamma^0$. Global probability conservation means that $\int_\Sigma d\sigma_\mu(x) \, j^\mu(x)$ does not depend on the choice of Cauchy surface $\Sigma \subset \R^4$.

For a multi-time wave function for $N$ Dirac particles, these notions can be generalized as follows. Local probability conservation can be expressed by a set of $N$ continuity equations,
\be
\partial_{x_k^\mu} \, j^{\mu_1 ... \mu_N}(x_1,...,x_N) = 0,~~k=1,2,...,N
\label{eq:continuity}
\ee
where $j^{\mu_1 ... \mu_N} = \overline{\psi} \gamma_1^{\mu_1} \cdots \gamma_N^{\mu_N} \psi$ denotes the Dirac tensor current.

Eqs. \eqref{eq:continuity} make it possible formulate a generalized version of the Born rule for all Cauchy surfaces. Let $n$ be the future-directed unit normal vector field at $\Sigma$. Then
\be
\rho(x_1,...,x_N) = \overline{\psi}(x_1,...,x_N) \gamma_1^{\mu_1} \cdots \gamma_N^{\mu_N} \psi(x_1,...,x_N) n_{\mu_1}(x_1) \cdots n_{\mu_N}(x_N)
\ee
defines the probability density for $N$ particles $k=1,...,N$ to cross $\Sigma$ at the locations $x_1,...,x_N \in \Sigma$. In fact, for theories with local interactions and finite propagation speed, it is possible to prove this rule using the usual Born rule in a distinguished frame \cite{generalized_born}.

The continuity equations \eqref{eq:continuity} imply global probability conservation in the sense that
\be
P(\psi, \Sigma) = \int_{\Sigma^N} d \sigma_{\mu_1}(x_1) \cdots d \sigma_{\mu_N}(x_N) \, \overline{\psi}(x_1,...,x_N) \gamma_1^{\mu_1} \cdots \gamma_N^{\mu_N} \psi(x_1,...,x_N)
\label{eq:probint}
\ee
does not depend on the choice of Cauchy surface $\Sigma$. In fact, this requires \eqref{eq:continuity} only on the set of space-like configurations $\mathscr{S} \subset \R^{4N}$, not necessarily on the entire configuration-spacetime $\R^{4N}$.
However, in the case that \eqref{eq:continuity} hold true on $\R^{4N}$, one finds that the generalized probability integral
\be
P(\psi, \Sigma_1, ..., \Sigma_N) = \int_{\Sigma_1 \times \cdots \times \Sigma_N} d \sigma_{\mu_1}(x_1) \cdots d \sigma_{\mu_N}(x_N) \, \overline{\psi}(x_1,...,x_N) \gamma_1^{\mu_1} \cdots \gamma_N^{\mu_N} \psi(x_1,...,x_N)
\label{eq:probint2}
\ee
is independent of the choice of $N$ (potentially different) Cauchy surfaces $\Sigma_1, ..., \Sigma_N$.

While $P(\psi,\Sigma)$ seems like the physically appropriate choice\footnote{The reason for this is that space-like configurations are the natural generalization of equal-time configurations. Non-space-like configurations can arise from multiple points on a single time-like (or light-like) world line. Thus, there is no physical reason to expect probability conservation on such configurations.} for the probability integral, we shall consider $P(\psi,\Sigma_1,...,\Sigma_N)$ as the more general notion in the following. This simplifies to investigate which local conservation laws follow from the global ones (see Sec. \ref{sec:results_local}).

\section{Results} \label{sec:results}

\subsection{Probability conservation on all Cauchy surfaces for retarded Green’s functions and retarded interaction kernels} \label{sec:results_ret}

As \eqref{eq:inteq} is an integral equation with a non-local interaction term on configuration space, we do not expect it to imply local conservation laws. This means that a method to prove global probability conservation without first establishing local probability conservation is required. Conveniently, the propagator $S$ of the Dirac equation provides such a method.

To see this, let $\psi^\free(x)$ be a solution of the Dirac eq. and $\Sigma, \Sigma'$ Cauchy surfaces. Then:
\begin{align}
P(\psi^\free,\Sigma) &= \int_\Sigma d\sigma_\mu(x) ~ \overline{\psi}^\free(x) \gamma^\mu \psi^\free(x)\nonumber\\
&\stackrel{\eqref{eq:freeevol}}{=} -i \int_\Sigma d\sigma_\mu(x) ~ \overline{\psi}^\free(x) \gamma^\mu \int_{\Sigma'} d \sigma_\nu(x') S(x-x')\gamma^\nu \psi^\free(x')\nonumber\\
&= \int_{\Sigma'} d \sigma_\nu(x') \left( -i \int_\Sigma d\sigma_\mu(x) \, \overline{\psi}^\free(x) \gamma^\mu S(x-x') \right) \gamma^\nu \psi^\free(x')\nonumber\\
&\stackrel{\eqref{eq:freeevol}}{=} \int_{\Sigma'} d \sigma_\nu(x') ~ \overline{\psi}^\free(x') \gamma^\nu \psi^\free(x') = P(\psi^\free,\Sigma').
\end{align}

Using the relation of retarded Green’s functions to the propagator $S$, we now prove our result.

\begin{proposition} \label{prop:result_ret}
Consider the integral equation \eqref{eq:inteq} with retarded Green’s functions, $G_k = S_k^\ret,~k=1,2$ \eqref{eq:sret}, and retarded interaction kernel \eqref{eq:retardedint}. Then for every solution $\psi$ of \eqref{eq:inteq} on $\R^4\times \R^4$, the probability integral $P(\psi,\Sigma_1,\Sigma_2)$ \eqref{eq:probint2} does not depend on the choice of Cauchy surfaces $\Sigma_1, \Sigma_2 \subset \R^4$.
\end{proposition}

\begin{proof}
Let $\psi^\free$ a solution of the free multi-time Dirac equations \eqref{eq:multitimedirac} and $\psi$ a solution of the integral equation \eqref{eq:inteq}. Our strategy is to decompose $P(\psi,\Sigma_1,\Sigma_2)$ as
\be
P(\psi,\Sigma_1,\Sigma_2) = P(\psi^\free,\Sigma_1,\Sigma_2) + P_1(\psi,\Sigma_1,\Sigma_2)
\ee
and to show that $P_1$ vanishes for the retarded interaction kernel \eqref{eq:retardedint} and all Cauchy surfaces $\Sigma_1,\Sigma_2$. We already know that the free Dirac equations \eqref{eq:multitimedirac} imply the continuity equations \eqref{eq:continuity}. Thus:
\be
P(\psi^\free,\Sigma_1,\Sigma_2) = P(\psi^\free,\Sigma_3,\Sigma_4)
\ee
for all Cauchy surfaces $\Sigma_1,\Sigma_2,\Sigma_3,\Sigma_4$, this allows us to deduce
\be
P(\psi,\Sigma_1,\Sigma_2) = P(\psi^\free,\Sigma_1,\Sigma_2) = P(\psi^\free,\Sigma_3,\Sigma_4) = P(\psi,\Sigma_3,\Sigma_4)
\ee
which is the claim.
The main work is to prove that $P_1(\psi,\Sigma_1,\Sigma_2)$ vanishes. Plugging the right hand side of Eq. \eqref{eq:inteq} (for $\psi$) and its adjoint \eqref{eq:inteq2} (for $\overline{\psi}$) into $P(\psi,\Sigma_1,\Sigma_2)$, we find, considering \eqref{eq:symmetryk}: 
\begin{align}
&P_1(\psi,\Sigma_1,\Sigma_2) = P(\psi,\Sigma_1,\Sigma_2) - P(\psi^\free,\Sigma_1,\Sigma_2) = \nonumber \\
&2 \Im \int_{\Sigma_1 \times \Sigma_2} d\sigma_\mu(x_1) \, d\sigma_\nu(x_2) \int d^4 x_1' \, d^4 x_2'~ \overline{\psi}(x_1',x_2') K(x_1',x_2') S_1(x_1'-x_1) S_2(x_2'-x_2) \nonumber\\
&~~~\times \theta(x_1^0 - {x_1'}^0) \theta(x_2^0 - {x_2'}^0) \gamma_1^\mu \gamma_2^\nu \, \psi^\free(x_1,x_2)\nonumber\\
&+ \int_{\Sigma_1 \times \Sigma_2} d\sigma_\mu(x_1) \, d\sigma_\nu(x_2) \int d^4 x_1' \, d^4 x_2'~ \overline{\psi}(x_1',x_2') K(x_1',x_2') S_1(x_1'-x_1) S_2(x_2'-x_2)\nonumber\\
&~~~\times \theta(x_1^0 - {x_1'}^0) \theta(x_2^0 - {x_2'}^0) \gamma_1^\mu \gamma_2^\nu \int d^4 x_1'' \, d^4 x_2'' \, S_1(x_1-x_1'') S_2(x_2-x_2'') \nonumber\\
&~~~\times \theta(x_1^0 - {x_1''}^0) \theta(x_2^0 - {x_2''}^0) K(x_1'',x_2'') \psi(x_1'',x_2'')\nonumber\\
&=: P_{1,1}(\psi,\Sigma_1,\Sigma_2) + P_{1,2}(\psi,\Sigma_1,\Sigma_2)
\label{eq:p1}
\end{align}
with $P_{1,1}$ and $P_{1,2}$ defined as the two summands of the equation in the order of appearance.

It is crucial that due to the simple relation \eqref{eq:sret} of $S^\ret$ with the propagator $S$, the propagators $S_1$ and $S_2$ appear in the equation which, together with the hypersurface integrals $\int_{\Sigma_1 \times \Sigma_2} d\sigma_\mu(x_1) \, d\sigma_\nu(x_2)$, can be used to evolve $\psi^\free$. In the first term $P_{1,1}$, we would like to use \eqref{eq:freeevol} in both $x_1$ and $x_2$. On first glance, this does not seem possible because $\theta(x_1^0 - {x_1'}^0) \theta(x_2^0 - {x_2'}^0)$ depends on the time variables $x_1^0$ and $x_2^0$ of the Cauchy surfaces $\Sigma_1$ and $\Sigma_2$, respectively. However, since the propagator $S$ of the Dirac equation has only support inside of and on the light cone, we can write:
\be
S(x'-x) \theta(x^0 - {x'}^0) = S(x'-x) \theta_\Sigma(x')~~~\text{and}~~~S(x-x') \theta(x^0 - {x'}^0) = S(x-x') \theta_\Sigma(x')
\label{eq:sret2}
\ee
where $\Sigma$ is a Cauchy surface that contains $x$ and
\be
\theta_\Sigma(x') = \begin{cases} 1 & \text{if } x' \in \text{past}(\Sigma)\\ 0 & \text{else}. \end{cases}
\ee
Here, $\text{past}(\Sigma) = \bigcup_{x\in\Sigma} \text{past}(x)$ denotes the part of space-time "below $\Sigma$". An important point is that $\theta_\Sigma(x')$ does not depend of $x$ (as long as $x \in \Sigma$). Using \eqref{eq:sret2}, we obtain:
\begin{align}
P_{1,1}(\psi,\Sigma_1,\Sigma_2) &= 2 \Im \int_{\Sigma_1 \times \Sigma_2} d\sigma_\mu(x_1) \, d\sigma_\nu(x_2) \int d^4 x_1' \, d^4 x_2'~ \overline{\psi}(x_1',x_2') K(x_1',x_2') \nonumber\\
&~~~\times S_1(x_1'-x_1) S_2(x_2'-x_2) \theta_{\Sigma_1}(x_1') \theta_{\Sigma_2}(x_2') \gamma_1^\mu \gamma_2^\nu \, \psi^\free(x_1,x_2).
\end{align}
It is now possible to exchange the integrals, yielding
\begin{align}
&P_{1,1}(\psi,\Sigma_1,\Sigma_2) = 2 \Im \int d^4 x_1' \, d^4 x_2'~ \overline{\psi}(x_1',x_2') K(x_1',x_2') \theta_{\Sigma_1}(x_1') \theta_{\Sigma_2}(x_2') \nonumber\\
&~~~~~\times \int_{\Sigma_1 \times \Sigma_2} d\sigma_\mu(x_1) \, d\sigma_\nu(x_2) \, S_1(x_1'-x_1) S_2(x_2'-x_2) \gamma_1^\mu \gamma_2^\nu \, \psi^\free(x_1,x_2).
\end{align}
This allows us to employ the propagation identity \eqref{eq:freeevol} twice to deduce:
\be
P_{1,1}(\psi,\Sigma_1,\Sigma_2) = - 2 \Im \int d^4 x_1' \, d^4 x_2'~ \overline{\psi}(x_1',x_2') K(x_1',x_2') \theta_{\Sigma_1}(x_1') \theta_{\Sigma_2}(x_2') \psi^\free(x_1',x_2').
\ee
Now we use the integral equation \eqref{eq:inteq} "backwards" to express $\psi^\free$ in terms of $\psi$:
\begin{align}
&P_{1,1}(\psi,\Sigma_1,\Sigma_2) =
- 2 \Im \int d^4 x_1' \, d^4 x_2'~ \overline{\psi}(x_1',x_2') K(x_1',x_2') \theta_{\Sigma_1}(x_1') \theta_{\Sigma_2}(x_2') \psi(x_1',x_2')\nonumber\\
&+ 2 \Im \, i \int d^4 x_1' \, d^4 x_2' \, \int d^4 x_1'' \, d^4 x_2''~ \overline{\psi}(x_1',x_2') K(x_1',x_2') \theta_{\Sigma_1}(x_1') \theta_{\Sigma_2}(x_2') \nonumber\\
&~~~\times \theta({x_1'}^0 - {x_1''}^0)\theta({x_2'}^0 - {x_2''}^0) S_1(x_1'-x_1'') S_2(x_2'-x_2'') K(x_1'',x_2'') \psi (x_1'',x_2'')\nonumber\\
&= 0 + 2 \Re \int d^4 x_1' \, d^4 x_2' \, \int d^4 x_1'' \, d^4 x_2''~ \overline{\psi}(x_1',x_2') K(x_1',x_2') \theta_{\Sigma_1}(x_1') \theta_{\Sigma_2}(x_2') \nonumber\\
&~~~\times \theta({x_1'}^0 - {x_1''}^0)\theta({x_2'}^0 - {x_2''}^0) S_1(x_1'-x_1'') S_2(x_2'-x_2'') K(x_1'',x_2'') \psi (x_1'',x_2'').
\end{align}
In order to conclude that the first term vanishes, we have used the symmetry of $K$ \eqref{eq:symmetryk}. Now we employ the identity $2\Re z = z + z^*$:
\begin{align}
& P_{1,1}(\psi,\Sigma_1,\Sigma_2) =\label{eq:p11}\\
&\int d^4 x_1' \, d^4 x_2' \, d^4 x_1'' \, d^4 x_2'' ~ \overline{\psi}(x_1',x_2') K(x_1',x_2') S_1(x_1'-x_1'')S_2(x_2'-x_2'') K(x_1'',x_2'') \psi(x_1'',x_2'')\nonumber\\
& \times \left[ \theta_{\Sigma_1}(x_1') \theta_{\Sigma_2}(x_2') \theta({x_1'}^0 - {x_1''}^0) \theta({x_2'}^0 - {x_2''}^0) + \theta_{\Sigma_1}(x_1'') \theta_{\Sigma_2}(x_2'') \theta({x_1''}^0 - {x_1'}^0) \theta({x_2''}^0 - {x_2'}^0)\right]
\nonumber
\end{align}
We compare this term to $P_{1,2}(\psi,\Sigma_1,\Sigma_2)$. Using \eqref{eq:sret2} and exchanging the order of the integrals, we obtain:
\begin{align}
&P_{1,2}(\psi,\Sigma_1,\Sigma_2) = \int d^4 x_1' \, d^4 x_2' \, d^4 x_1'' \, d^4 x_2''~ \overline{\psi}(x_1',x_2') K(x_1',x_2') \nonumber\\
&~~~\times \int_{\Sigma_1 \times \Sigma_2} d\sigma_\mu(x_1) \, d\sigma_\nu(x_2) S_1(x_1'-x_1) S_2(x_2'-x_2) \gamma_1^\mu \gamma_2^\nu S_1(x_1-x_1'') S_2(x_2-x_2'') \nonumber\\
&~~~\times \theta_{\Sigma_1}(x_1') \theta_{\Sigma_2}(x_2') \theta_{\Sigma_1}(x_1'') \theta_{\Sigma_2}(x_2'') K(x_1'',x_2'') \psi(x_1'',x_2'').
\end{align}
This allows us to utilize the composition property \eqref{eq:compositionproperty} for the propagators twice, yielding
\begin{align}
P_{1,2}(\psi,\Sigma_1,\Sigma_2) &= - \int d^4 x_1' \, d^4 x_2' \, d^4 x_1'' \, d^4 x_2''~ \overline{\psi}(x_1',x_2') K(x_1',x_2') S_1(x_1'-x_1'') S_2(x_2'-x_2'') \nonumber\\
&~~~\times K(x_1'',x_2'') \psi(x_1'',x_2'') \theta_{\Sigma_1}(x_1') \theta_{\Sigma_2}(x_2') \theta_{\Sigma_1}(x_1'') \theta_{\Sigma_2}(x_2'').
\label{eq:p12}
\end{align}
Comparing \eqref{eq:p11} and \eqref{eq:p12}, we find that $P_{1,1}$ and $P_{1,2}$  cancel if:
\begin{align}
&\theta_{\Sigma_1}(x_1') \theta_{\Sigma_2}(x_2') \theta_{\Sigma_1}(x_1'') \theta_{\Sigma_2}(x_2'') ~\stackrel{!}{=} \label{eq:criticalcond0}\\
&\theta_{\Sigma_1}(x_1') \theta_{\Sigma_2}(x_2') \theta({x_1'}^0 - {x_1''}^0) \theta({x_2'}^0 - {x_2''}^0)
+ \theta_{\Sigma_1}(x_1'') \theta_{\Sigma_2}(x_2'') \theta({x_1''}^0 - {x_1'}^0) \theta({x_2''}^0 - {x_2'}^0).\nonumber
\end{align}
To make further progress, note that
\be
\theta_{\Sigma_k}(x_k') \theta({x_k'}^0 - {x_k''}^0) = \theta_{\Sigma_k}(x_k') \theta_{\Sigma_k}(x_k'') \theta({x_k'}^0 - {x_k''}^0)
\ee
as $x_k'' \in \text{past}(x_k')$ and $x_k' \in \text{past}(\Sigma_k)$ together imply that $x_k'' \in \text{past}(\Sigma_k)$. Thus, every term in \eqref{eq:criticalcond0} contains the factor $\theta_{\Sigma_1}(x_1') \theta_{\Sigma_2}(x_2')\theta_{\Sigma_1}(x_1'') \theta_{\Sigma_2}(x_2'')$ and \eqref{eq:criticalcond0} reduces to:
\be
1 ~\stackrel{!}{=}~ \theta({x_1'}^0 - {x_1''}^0) \theta({x_2'}^0 - {x_2''}^0) + \theta({x_1''}^0 - {x_1'}^0) \theta({x_2''}^0 - {x_2'}^0).
\label{eq:criticalcond1}
\ee
In general, this condition does not hold. However, we have not used the special structure of the retarded interaction kernel $K^\ret$ \eqref{eq:retardedint} yet. Instead of \eqref{eq:criticalcond1}, we will show that
\begin{align}
&K^\ret(x_1',x_2') S_1(x_1'-x_1'') S_2(x_2'-x_2'') K^\ret(x_1'',x_2'')\nonumber\\
&\times \left[ 1 - \theta({x_1'}^0 - {x_1''}^0) \theta({x_2'}^0 - {x_2''}^0) - \theta({x_1''}^0 - {x_1'}^0) \theta({x_2''}^0 - {x_2'}^0)\right] = 0.
\label{eq:criticalcond}
\end{align}
In fact, including the first line adds several geometric conditions on those tuples $(x_1',x_2',x_1'',x_2'')$ which may actually contribute to $P_{1,1}(\psi,\Sigma_1,\Sigma_2) + P_{1,2}(\psi,\Sigma_1,\Sigma_2)$.
\begin{enumerate}
\item[(i)] ${x_2'}^0 = {x_1'}^0 - |\vx_1' - \vx_2'|$ (because of $K^\ret(x_1',x_2')$),
\item[(ii)] ${x_2''}^0 = {x_1''}^0 - |\vx_1'' - \vx_2''|$ (because of $K^\ret(x_1'',x_2'')$),
\item[(iii)] $|{x_1'}^0 - {x_1''}^0| \geq |\vx_1' - \vx_1''|$ (because of the support of $S_1$),
\item[(iv)] $|{x_2'}^0 - {x_2''}^0| \geq |\vx_2' - \vx_2''|$ (because of the support of $S_2$).
\end{enumerate}
We now show that the above conditions only allow for the following two cases:
\begin{center}
	1. ${x_k'}^0 > {x_k''}^0$ for $k=1,2$  ~~~or~~~ 2. ${x_k''}^0 > {x_k'}^0$ for $k=1,2$.
\end{center}
The logic behind this claim is that if true, $K^\ret(x_1',x_2') S_1(x_1'-x_1'') S_2(x_2'-x_2'') K^\ret(x_1'',x_2'') \neq 0$ implies $\theta({x_1'}^0 - {x_1''}^0) \theta({x_2'}^0 - {x_2''}^0) + \theta({x_1''}^0 - {x_1'}^0) \theta({x_2''}^0 - {x_2'}^0) = 1$. This, in turn, means that the square bracket in \eqref{eq:criticalcond} vanishes. Hence, \eqref{eq:criticalcond} is always satisfied, implying $P_1(\psi,\Sigma_1,\Sigma_2) = P_{1,1}(\psi,\Sigma_1,\Sigma_2) + P_{1,2}(\psi,\Sigma_1,\Sigma_2) = 0$, and thus probability conservation.

\begin{figure}
\centering
\includegraphics[width=0.21\textwidth]{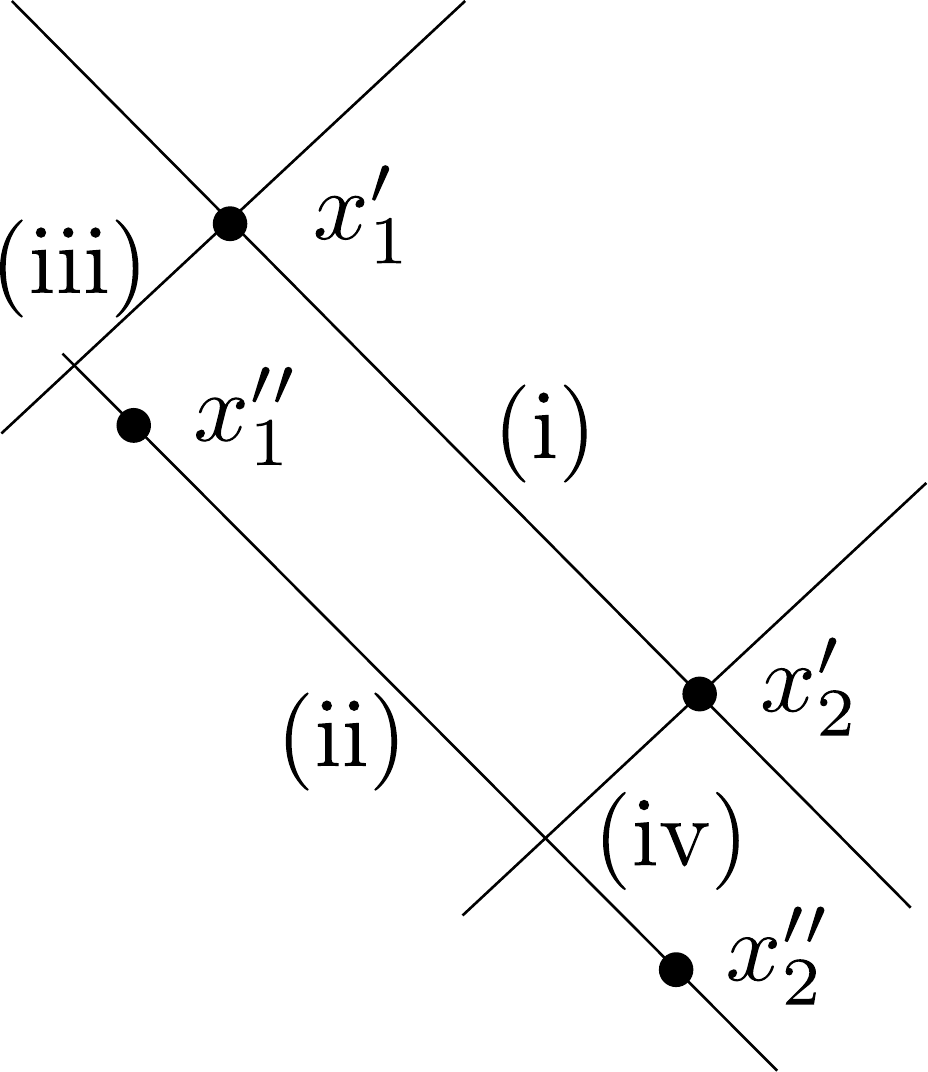}
\caption{Illustration of the proof of \eqref{eq:criticalcond}. Begin with $x_1'$. By (i), $x_2'$ has to lie on the past light cone of $x_1'$. Next, consider $x_1''$. Assume that ${x_1'}^0 >{x_1''}^0$. By (iii), $x_1''$ has to be in the past of $x_1'$. The final point $x_2''$ has to lie on the past light cone of $x_1''$ by (ii). Moreover, by (iv), $x_2''$ has to be light-like or time-like to $x_2'$. As one can see, there is no other option for $x_2''$ than to lie in the past of $x_2'$, which implies ${x_2'}^0 > {x_2''}^0$.} 
\label{fig:light_cone_construction}
\end{figure}

We begin with the first case and demonstrate that ${x_1'}^0 > {x_1''}^0$ implies ${x_2'}^0 > {x_2''}^0$ (see Fig. \ref{fig:light_cone_construction} for geometrical intuition). 
Consider the difference ${x_2'}^0 - {x_2''}^0$. Combining (i) and (ii) yields:
\be
{x_2'}^0 - {x_2''}^0 = {x_1'}^0 - {x_1''}^0 - |\vx_1' - \vx_2'| + |\vx_1'' - \vx_2''|.
\ee
Now, because of (iii) and ${x_1'}^0 > {x_1''}^0$, this implies:
\be
{x_2'}^0 - {x_2''}^0 \geq |\vx_1' - \vx_1''| - |\vx_1' - \vx_2'| + |\vx_1'' - \vx_2''|.
\ee
Next, we use the triangle inequality twice to obtain:
\be
	|\vx_1' - \vx_2'| ~\geq~ |\vx_1' - \vx_1''| + |\vx_1''-\vx_2'| ~\geq~ |\vx_1' - \vx_1''| + |\vx_1'' - \vx_2''| + |\vx_2'' - \vx_2'|.
\ee
Together with the previous inequality, this gives us:
\be
{x_2'}^0 - {x_2''}^0 ~\geq~ - |\vx_2' - \vx_2''| ~~~
\Leftrightarrow ~~~  {x_2''}^0 - {x_2'}^0 ~\leq~ |\vx_2' - \vx_2''|.
\label{eq:casedifferentiator}
\ee
Note that relation (iv) implies that there are only two cases:
\begin{center}
	(a) ${x_2'}^0 - {x_2''}^0 \geq |\vx_2' - \vx_2''|$ ~~~ or ~~~ (b) ${x_2''}^0 - {x_2'}^0 \geq |\vx_2' - \vx_2''|$.
\end{center}
The crucial point now is that \eqref{eq:casedifferentiator} contradicts (b) while being compatible with (a). As (a) and (b) are mutually exclusive, this implies (a) which in particular establishes ${x_2'}^0 \geq {x_2''}^0$.

The fact that ${x_2'}^0 > {x_2''}^0$ implies ${x_1'}^0 > {x_1''}^0$ follows from the same consideration and the fact that ${x_1''}^0 > {x_1'}^0$ is equivalent to ${x_2''}^0 > {x_2'}^0$ by exchanging $x_k'$ with $x_k''$ for $k=1,2$. \qed
\end{proof}

\paragraph{Remark.}
In special cases, the Feynman propagator $S^F$ can also be used to propagate $\psi^\free$ (see \cite[chap. 6.1]{greiner_qed}), e.g., for propagating positive energy wave functions towards the future. However, it does not vanish outside of the light cone. Since this property of $S^\ret$ is an integral part of the proof, the latter cannot be extended to the case of $G_k = S^F_k$. This suggests that probability conservation does not hold for the Bethe-Salpeter equation, in agreement with the literature.

\subsection{Asymptotic probability conservation for symmetric Green’s fns.} \label{sec:results_sym}

One may wonder if probability conservation can be established for different classes of interaction kernels besides retarded ones. We now prove such a result, albeit a weaker one, for the case of symmetric Green’s functions $G_k = S_k^\sym$ with
\be
S^\sym(x-x’) = \frac{1}{2} \varepsilon(x^0-{x’}^0)S(x-x’)
\label{eq:ssym}
\ee
where $\varepsilon(y) = +1$ if $y \geq 0$ and $\varepsilon(y) = -1$ else.

\begin{proposition} \label{prop:result_sym}
Consider the integral equation \eqref{eq:inteq} with symmetric Green’s functions, $G_k = S_k^\sym,~k=1,2$ \eqref{eq:ssym}, and interaction kernels with the matrix symmetry property \eqref{eq:symmetryk}. Let $\Sigma_t$ be an equal-time surface in any given Lorentz frame. Then for every solution $\psi$ of \eqref{eq:inteq} on $\R^4\times \R^4$, the following statement \emph{(asymptotic probability conservation)} holds true:
\be
\lim_{t\rightarrow -\infty} P(\psi, \Sigma_t) = \lim_{t\rightarrow +\infty} P(\psi, \Sigma_t).
\ee
\end{proposition}

\begin{proof}
We proceed similarly to the retarded case. The main difference is that the Heaviside functions $\theta$ get replaced by $\frac{1}{2} \varepsilon$. As we focus on equal-time surfaces $\Sigma_t$, the propagation identities \eqref{eq:freeevol} and \eqref{eq:compositionproperty} can be used directly. One obtains:
\begin{align}
&P(\psi,\Sigma_t) = P(\psi^\free,\Sigma_t) + P_1(\psi,\Sigma_t),~~~\text{with}
\label{eq:psym}\\
&P_1(\psi,\Sigma_t) =\frac{1}{16} \int d^4 x_1' \, d^4 x_2'\, d^4 x_1'' \, d^4 x_2''~ \overline{\psi}(x_1',x_2') K(x_1',x_2') S_1(x_1'-x_1'') S_2(x_2'-x_2'')\nonumber\\
&~~~\times K(x_1'',x_2'') \psi(x_1'',x_2'')\left[ - \varepsilon(t-{x_1'}^0) \varepsilon(t-{x_2'}^0)\varepsilon(t-{x_1''}^0) \varepsilon(t-{x_2''}^0)\right.\nonumber\\
&~~~\left. + \left( \varepsilon(t-{x_1'}^0) \varepsilon(t-{x_2'}^0) + \varepsilon(t-{x_1''}^0) \varepsilon(t-{x_2''}^0) \right)\varepsilon({x_1'}^0-{x_1''}^0) \varepsilon({x_2'}^0-{x_2''}^0) \right]
\label{eq:p1sym}
\end{align}
Only the square bracket depends on $t$ -- and is not constant in $t$, as can be seen by case differentiation. However, taking $t\rightarrow \pm \infty$ both leads to the same result
\begin{align}
	&\lim_{t \rightarrow \pm \infty}[...] = -1 + 2\varepsilon({x_1'}^0-{x_1''}^0) \varepsilon({x_2'}^0-{x_2''}^0).\\
	\text{Thus, we find:}~~~ &\lim_{t\rightarrow -\infty} P_1(\psi,\Sigma_t) = \lim_{t\rightarrow +\infty} P_1(\psi,\Sigma_t).
\label{eq:p1symlimit}
\end{align}
Interestingly, the term $P_1$ is non-zero here, in contrast to the retarded case.\footnote{In the symmetric case $\psi^\free$ does not need to agree with $\psi$ in the infinite past nor in the infinite future \cite{int_eq_curved}.} Nevertheless, together with probability conservation for $\psi^\free$, \eqref{eq:p1symlimit} allows us to deduce:
\begin{align}
\lim_{t\rightarrow -\infty}P(\psi,\Sigma_t) = P(\psi^\free,-\infty) + P_1(\psi,-\infty)= P(\psi^\free,+\infty) + P_1(\psi,+\infty) = \lim_{t\rightarrow +\infty}P(\psi,\Sigma_t)P(\psi,\Sigma_t).
\end{align}
\qed
\end{proof}

\vspace{-1cm}

\subsection{Implications for local conservation laws} \label{sec:results_local}

\begin{figure}
\centering
\includegraphics[width=0.3\textwidth]{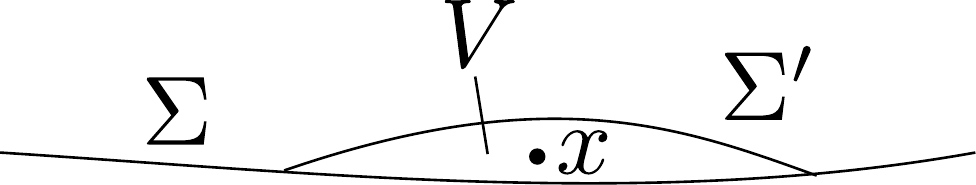}
\caption{Geometric construction for proving that global probability conservation on all Cauchy surfaces implies local probability conservation in the single-particle case.}
\label{fig:bubbleconstruction}
\end{figure}

It is a well-known fact that in the single-particle case global probability conservation on all Cauchy surfaces, $\int_{\Sigma} d\sigma_\mu(x) \, j^\mu(x) = \int_{\Sigma'} d\sigma_\mu(x) \, j^\mu(x)~\forall \, \Sigma, \Sigma'$, implies local probability conservation. This can be shown by constructing a small volume $V$ around a space-time point $x$ which is enclosed between two otherwise overlapping Cauchy surfaces $\Sigma, \Sigma'$ (see Fig. \ref{fig:bubbleconstruction}) and then using the divergence theorem. As the enclosing volume $V$ can be made arbitrarily small, it follows that $\partial_{x^\mu} \, j^\mu(x) = 0$.

Prop. \ref{prop:result_ret} establishes global probability conservation on all Cartesian products $\Sigma_1 \times \Sigma_2$ of Cauchy surfaces, enabling a similar reasoning. Applying the argument to $x_k$ and $\Sigma_k$ yields
\be
\int_{\Sigma_{3-k}}\!\!\! d\sigma_{\mu_{3-k}}(x_{3-k}) \, \partial_{x_k^{\mu_k}} \left[\overline{\psi}(x_1,x_2)\gamma_1^{\mu_1} \gamma_2^{\mu_2} \psi(x_1,x_2) \right] = 0~~\forall \Sigma_{3-k},~~k=1,2.
\label{eq:bubbleargument1}
\ee
Now, we apply the argument another time for $x_{3-k}$ and $\Sigma_{3-k}$ and obtain the result:

\begin{proposition} \label{prop:result_local}
Let $\psi$ be a solution of the integral equation \eqref{eq:inteq} with $G_k = G_k^\ret,~k=1,2$ and $K = K^\ret$ \eqref{eq:retardedint}. Then \eqref{eq:bubbleargument1} is satisfied. In addition, we have:
\be
\partial_{x_1^\mu} \partial_{x_2^\nu} \, \overline{\psi}(x_1,x_2) \gamma_1^\mu \gamma_2^\nu \psi(x_1,x_2) = 0.
\label{eq:localproperty}
\ee
\end{proposition}

\vspace{-0.6cm}
\paragraph{Remarks.}
\begin{enumerate}
\item Eq. \eqref{eq:bubbleargument1} for $k=1,2$ implies global probability conservation on all space-like Cauchy surfaces (as can be seen using the divergence theorem), i.e., both notions are equivalent. Eq. \eqref{eq:localproperty}, however, is weaker. Its physical meaning is not clear.
\item \eqref{eq:bubbleargument1} for $k=1,2$ \emph{does not imply local probability conservation}. The reason is that even though \eqref{eq:bubbleargument1} holds for all $\Sigma_2$, one cannot conclude that the integrand vanishes, as it might take negative values. This leads to the conclusion that \emph{global probability conservation on all Cartesian products $\Sigma \times \Sigma$ is not equivalent to local probability conservation.}
\item To make it even clearer that local probability conservation does not hold, we calculate the four-divergence of the tensor current with respect to $x_1$:
\begin{align}
\partial_{x_1^\mu}\left[\overline{\psi}\gamma_1^\mu \gamma_2^\nu \psi\right]\! (x_1,x_2) &= -2 \Re \Big[ \overline{\psi}(x_1,x_2) \gamma_2^\nu \int  \! d^4 x_2'~S_2^\ret(x_2-x_2') K^\ret(x_1,x_2')\psi(x_1,x_2')\Big].
\label{eq:dx1mujmunu}
\end{align}
Despite the spatio-temporal restrictions which $S_2^\ret$ and $K_2^\ret$ imply and the fact that it might be sufficient to restrict the derivation to space-like configurations $(x_1,x_2) \in \mathscr{S}$, the right hand side of \eqref{eq:dx1mujmunu} does not vanish in general.
\item Concerning the physical meaning of Eqs. \eqref{eq:bubbleargument1} for $k=1,2$, one can rewrite them as:
\begin{align}
	&\partial_{k,\mu} j_k^\mu(x_k,\Sigma) = 0~~\forall \, \Sigma,~k=1,2\\
\text{where}~~~&j_k^\mu(x_k,\Sigma) = \int_{\Sigma} d\sigma_\nu(x_{3-k})~\overline{\psi}(x_1,x_2) \gamma_1^\mu \gamma_2^\nu \psi(x_1,x_2),~k=1,2
\label{eq:singleparticlecurrents}
\end{align}
One could try to use these currents to naively calculate the probability for the position of one particle, disregarding the position of the other. Viz, one might guess that the probability for particle $k$ to be found in a small volume $d\sigma(x_k)$ around $x_k \in \Sigma$ is given by
\be
\mathbb{P}(x_k \in d\sigma(x_k)) = j_k^\mu(x_k,\Sigma) n_\mu(x_k) d\sigma(x_k).
\ee
where $n_\mu(x)$ is the future-directed unit normal vector field at $x \in \Sigma$.
\end{enumerate}

\section{Conclusion} \label{sec:conclusion}

Here it was shown that global probability conservation on all Cauchy surfaces holds for certain classes of multi-time integral equations.  To make progress, it was necessary to deviate from the conventional wisdom about relativistic quantum-mechanical integral equations given by the theory about the Bethe-Salpeter equation. The strongest result was obtained for retarded Green's functions $G_k=S_k^\ret,~k=1,2$ and retarded interaction kernels $K=K^\ret$ \eqref{eq:retardedint}.
While retarded interactions are common in classical electrodynamics, a word of caution seems in order here. Using $K^\ret(x_1,x_2)$ implies that $x_1^0 > x_2^0$ has to hold for a configuration $(x_1,x_2)$ to contribute to the interaction term. This seems unnatural as it breaks the symmetry between the particle labels. Related to this, for $S_k^\ret$, $K^\ret$ and on space-like configurations $(x_1,x_2)\in \mathscr{S}$ one has 
\begin{align}
(-i\gamma_2^\nu \partial_{x_2^\nu}+m_2) \psi(x_1,x_2) &= i\int d^4 x_1' \, S_1^\ret(x_1-x_1')K^\ret(x_1',x_2)\psi(x_1',x_2) = 0
\label{eq:d2psiretspacelike}
\end{align}
and therefore $\partial_{x_2^\nu} \, \overline{\psi}(x_1,x_2)\gamma_1^\mu \gamma_2^\nu \psi(x_1,x_2) = 0~~~\forall \, (x_1,x_2)\in \mathscr{S}$. 
Eq. \eqref{eq:d2psiretspacelike} can be shown as follows. Consider the conditions for a term $\psi(x_1',x_2)$ to contribute to the integral. On the one hand, $K^\ret(x_1',x_2) \neq 0$ implies that $x_2$ lies on the past light cone of $x_1'$. On the other hand, $S_2^\ret(x_1 - x_1') \neq 0$ implies that $x_1$ lies in the future of $x_1'$. This is, however, incompatible with $(x_1,x_2) \in \mathscr{S}$.

Thus, on $\mathscr{S}$, one can take \eqref{eq:d2psiretspacelike} to express that particle 2 is moving freely (while particle 1 is not) and interpret the interaction term in \eqref{eq:inteq} as a single-sided action of particle 2 on particle 1.\footnote{On configurations $(\R^4\times \R^4) \setminus \mathscr{S}$, however, which the probability integral on $\Sigma_1\times \Sigma_2$ for $\Sigma_1 \neq \Sigma_2$ uses, one in general has $(-i\gamma_2^\nu \partial_{x_2^\nu}+m_2) \psi(x_1,x_2) \neq 0$.} This strengthens the concern that this type of interaction is physically not natural.
In my opinion, the resulting dynamics represents a toy example and a first step towards a more natural result in the future, e.g., for $K^\sym$ \eqref{eq:lightconeint}.

In view of this situation, the second result, asymptotic probability conservation for the integral equation with symmetric Green’s functions, $G_k = S_k^\sym,~k=1,2$ seems particularly important. While weaker than probability conservation on all Cauchy surfaces, it holds for \emph{arbitrary interaction kernels} respecting the basic symmetry property \eqref{eq:symmetryk}.

Moreover, it was shown that global probability conservation on all Cartesian products $\Sigma_1\times \Sigma_2$ of Cauchy surfaces is equivalent to the semi-local property \eqref{eq:bubbleargument1}. Local probability conservation, however, does seem not hold.\footnote{Eq. \eqref{eq:d2psiretspacelike} implies $\partial_{x_2^\nu} \overline{\psi}(x_1,x_2) \gamma_1^\mu \gamma_2^\nu \psi(x_1,x_2) = 0$ only on space-like configurations $(x_1,x_2) \in \mathscr{S}$ while $\partial_{x_1^\mu} \overline{\psi}(x_1,x_2) \gamma_1^\mu \gamma_2^\nu \psi(x_1,x_2) \neq 0$  in general on $\R^4 \times \R^4$.} While perfectly logical, this fact may seem surprising since in the single-particle case global probability conservation on all Cauchy surfaces $\Sigma$ is equivalent to local probability conservation. Thus, we found that, \emph{in the $N$-particle case, local probability conservation is stronger than global probability conservation on all (Cartesian products of) $N$ Cauchy surfaces}.

In the future, it would be interesting to investigate if physical meaning can be given to the semi-local conservation law \eqref{eq:bubbleargument1}. For example, one may wonder if these properties are helpful to construct relativistic Bohmian laws of motion along the lines of \cite{hbd} or \cite{opposite_arrows}, or to any other theory in the foundations of quantum mechanics which avoids the measurement problem. That would be in Detlef D\"urr’s spirit.

During the time when he was my PhD adviser, Detlef D\"urr expressed that he thought further progress in the foundations of relativistic quantum theory required a clear and simple, mathematically solid underlying equation (free of the divergences that plague quantum field theory), be it only for a toy example. Once found, one could hope that such an equation would provide further guidelines for constructing a relativistic law of motion for Bohmian particles, as the Schr\"odinger and Dirac equations do in non-relativistic QM and relativistic single-particle QM, respectively. Perhaps, the integral equation discussed here can provide a starting point for such a consideration.

Overall, what has been achieved? I would say, a new way of constructing an interacting, relativistic equation which is now demonstrably compatible with global probability conservation. Do the examples provided constitute a full theory of relativistic quantum physics? Certainly not. Alas, one may conclude with a twist on the words of Feynman from the introduction:

\begin{quote}
We now have a satisfactory relativistic quantum-mechanical model, one that doesn’t agree with nature, but, at least, agrees with the logic that the sum of probability of all alternatives has to be 100\%.
\end{quote}

\paragraph{Disclaimer.}
This article reflects my personal work and does not represent scientific standpoints of my employer (Marvel Fusion GmbH).

\paragraph{Acknowledgments.}
I would like to thank Siddhant Das and Roderich Tumulka for helpful discussions as well as Sascha Lill and an anonymous referee for constructive comments on the manuscript. Special thanks go to Markus N\"oth for valuable remarks and a careful reading of an earlier version of the argument.


\end{document}